%% file: nips_2019.tex
\documentclass{article}

   \PassOptionsToPackage{numbers, compress}{natbib}
\usepackage[final]{neurips_2019_ml4ad}

\usepackage{graphicx}
\usepackage{float}

\usepackage{xcolor}

\usepackage{algorithmic}
\usepackage{algorithm}

\usepackage[utf8]{inputenc} 
\usepackage[T1]{fontenc}    
\usepackage{hyperref}       
\usepackage{url}            
\usepackage{booktabs}       
\usepackage{amsfonts}       
\usepackage{nicefrac}       
\usepackage{microtype}      
\usepackage{times}
\usepackage{helvet}
\usepackage{courier}
\usepackage[algo2e]{algorithm2e}
\usepackage{amsmath, bm}
\usepackage{amssymb}
\usepackage{amsthm}
\usepackage{amsfonts}
\newtheorem{theorem}{Theorem}

\newtheorem{definition}{Definition}
\SetKwInput{KwInput}{Input}
\SetKwInput{KwOutput}{Output}
\setcitestyle{square}

\title{Spatial Influence-aware Reinforcement Learning for Intelligent Transportation System}

\author{Wenhang Bao$^+$, 
       Xiao-Yang Liu$^*$\\
       $^*$Electrical Engineering, Columbia University,\\
       $^+$Department of Statistics, Columbia University,\\
       Emails: \{WB2304, XL2427\}@columbia.edu   
       }

\begin{document}

\maketitle

\begin{abstract}


Intelligent transportation systems (ITSs) are envisioned to be crucial for smart cities, which aims at improving traffic flow to improve the life quality of urban residents and reducing congestion to improve the efficiency of commuting. However, several challenges need to be resolved before such systems can be deployed, for example, conventional solutions for Markov decision process (MDP) and single-agent Reinforcement Learning (RL) algorithms suffer from poor scalability, and multi-agent systems suffer from poor communication and coordination. In this paper, we explore the potential of mutual information sharing, or in other words, spatial influence based communication, to optimize traffic light control policy. First, we mathematically analyze the transportation system. We conclude that the transportation system does not have stationary Nash Equilibrium, thereby reinforcement learning algorithms offer suitable solutions. Secondly, we describe how to build a multi-agent Deep Deterministic Policy Gradient (DDPG) system with spatial influence and social group utility incorporated. Then we utilize the grid topology road network to empirically demonstrate the scalability of the new system. We demonstrate three types of directed communications to show the effect of directions of social influence on the entire network utility and individual utility. Lastly, we define ``selfish index'' and analyze the effect of it on total group utility.
  
\end{abstract}

\input{Sections/Introduction.tex}

\input{Sections/ProblemDescription.tex}

\input{Sections/DRL.tex}

\input{Sections/Performance.tex}

\input{Sections/Conclusion.tex}


\iftrue

\fi
\end{document}

%% file: Sections/Introduction.tex
\section{Introduction}
\label{sect:introduction}

Emerging intelligent transportation systems (ITSs) \cite{vazifeh2018addressing,zhu2018online,zhu2016public,zhu2018joint} are expected to play an instrumental role in improving traffic flow, thus optimizing fuel efficiency, reducing delays and enhancing the general driving experience. ITSs are designed to resolve traffic congestion, which is an exceedingly complex and important issue faced by metropolitan areas around the world, as a result of global urbanization. The urbanization process makes traffic a serious problem in the urban area, as there are so many commuting objects and vehicles. Street interactions in dense urban areas can be critical bottlenecks in urban road networks, which affect commuters' efficiency. The competing nature of different intelligent routing applications is not aiming at improving the efficiency of the city as a whole. They optimize the routing paradigm for their customers locally, which might not be the global optimal solution.

ITSs are expected to resolve these issues, as advanced communication and computing technologies are developed to allow efficient information sharing \cite{tampuu2017multiagent,foerster2018multi} among commuting objects, making intelligent routing and traffic light control possible \cite{alam2016introduction,lv2015traffic}. Current traffic light systems are typically hard-coded based on investigation results or pure experience. The light control policy is optimized based on historical data and adapted according to daily patterns or drivers' feedback. The opportunity of using Artificial Intelligence (AI) for adaptation to real-time conditions, e.g. through detection wires in the pavement, tends to be fairly rudimentary. Evolving technologies offer the option of using information collected by cameras, including fine-grained knowledge of the positions and speeds of the vehicles. Such comprehensive real-time information can be leveraged to improve traffic flow through more agile traffic light control systems. While the potential benefits are immense, so are the technical challenges that arise in solving such real-time actuation problems on an unprecedented scale in terms of intrinsic complexity, geographic range, and number of objects involved.  

Under suitable assumptions, the problem of optimal dynamic traffic control may be formulated as a Markov decision process (MDP) \cite{onori2016dynamic,puterman2014markov,ross2014introduction}. The MDP framework provides a rigorous notion of optimally along the basis for computational techniques such as value iteration, policy iteration or dynamic programming. However, an MDP formulation may not always be satisfied in reality. The knowledge of relevant environmental parameters may not be available. Also, the environment is not stationary. Moreover, in terms of computational cost, an MDP approach suffers from the curse of dimensionality, resulting in excessively large state spaces in the realistic traffic system. RL algorithms overcome some of these limitations and have been previously considered in the context of optimal dynamic traffic light control. Researchers tried to use single-agent algorithms, for example, DDPG \cite{lillicrap2015continuous}, to resolve this issue, but these methods are still prone to prohibitively large state spaces and action spaces, implying poor scalability beyond a single-interaction scenario.

In this paper, we explore the potential of using multi-agent reinforcement learning algorithms, particularly Multi-agent Deep Deterministic Policy Gradients (MADDPG) \cite{lowe2017multi}, to optimize real-time traffic light control policies in large-scale systems \cite{chu2019multi}. First, we formulate the ITSs environment and the optimization goal. Then we analyze the properties of the Transportation system. We provide mathematical proof that stationary Nash Equilibrium can not be achieved in this case. Therefore, Reinforcement Learning algorithms are needed to resolve it. Thirdly, we consider a grid network topology with multiple rows and columns to examine the scalability properties of MADDPG algorithms and their performance in the presence of highly complex interactions created by the flow of vehicles along the main artery. Under this scenario, the influencing network would have complexity similar to a real city. We develop inward, outward and fully connected spacial influencing flows, to verify the effect of influencing directions in optimizing traffic signal control. We conclude that the directions of the spatial influence would not affect total network utility but affect individual utility. Finally, we define and investigate the ``selfish index'' of agents to check its influence on group utility. 

The remainder of the paper is organized as follows. We first present a detailed traffic environment description and problem statement. Then we provide a specification of our customized deep reinforcement learning algorithm MADDPG for grid topology road network with border interactions. The properties of the Intelligent Transportation System are also discussed. Also, we empirically evaluate the performance of the proposed spatial influence based MADDPG algorithms and illustrate the scalability of this algorithm. Social group utility optimization through reward function adjustment is also investigated. Finally, we conclude this paper and point out some future directions.

%% file: Sections/ProblemDescription.tex
\section{Problem Description}
\label{sect:problem}

We model the road intersections and formulate our problem similar to  \cite{liu2018deep}. We consider a multi-agent reinforcement learning environment where each agent is responsible for one intersection. Our central aim here is to design a system and explore the performance and scalability of multi-agent reinforcement learning algorithms in optimizing real-time traffic light control policies, rather than design a productive level policy for a specific location. We adopt a discrete time formulation to simplify the description and allow the direct application of conventional MDP techniques for comparison. The methods can be easily converted to continuous time operation as well.

\subsection{Traffic Model}
We consider the simplest meaningful intersection setup with bidirectional traffic flows. Regardless of the number of intersections, all intersections follow the same settings. Vehicles are coming from $4$ directions as $4$ queues. We use $X_{ni}(t), i=1,2,3,4$, to denote the number of vehicles of traffic flow $i$ waiting to pass the intersection $n$ at time $t$ and $Y_n(t) \in \{0,1,2,3\}$ to indicate the configuration of the traffic lights for intersection $n$ at time $t$. The traffic light $L$ has four configurations:
\begin{itemize}
    \item``0'': green light for flow 1 and hence red light for flow 2;
    \item``1'': yellow light for flow 1 and hence red light for flow 2;
    \item``2'': green light for flow 2 and hence red light for flow 1;
    \item``3'': yellow light for flow 2 and hence red light for flow 1.
\end{itemize}

\begin{figure}
\begin{tabular}{cc}
  \includegraphics[width=60mm]{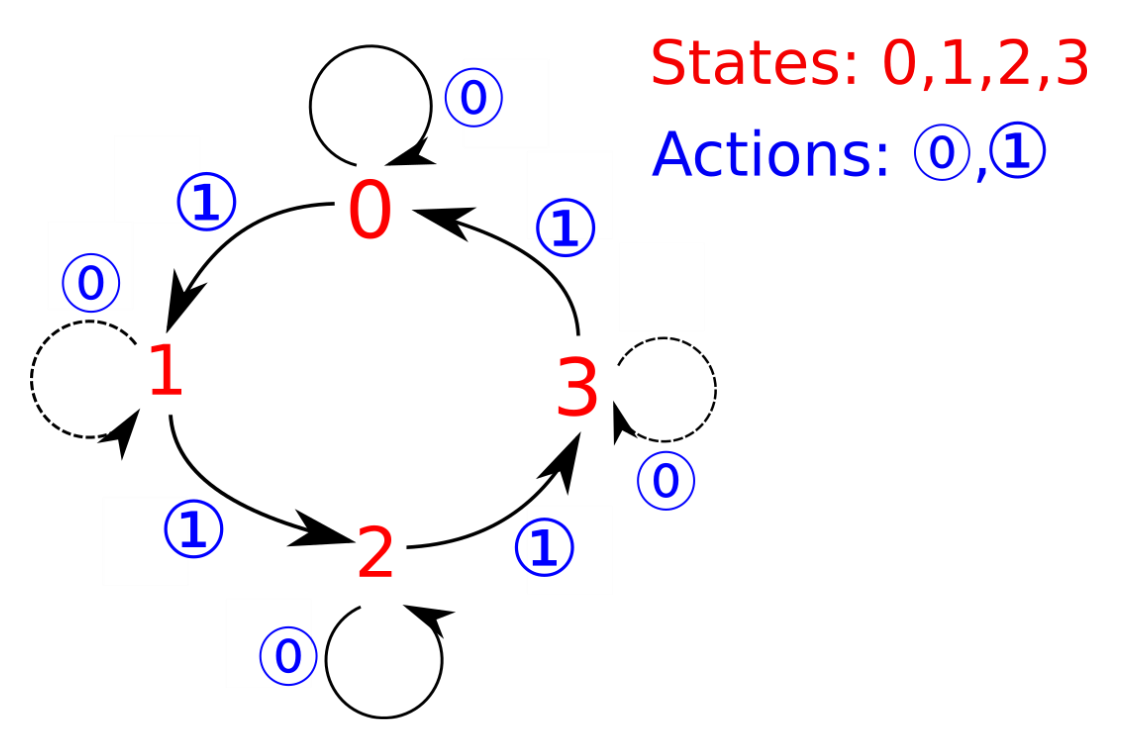} &   \includegraphics[width=60mm]{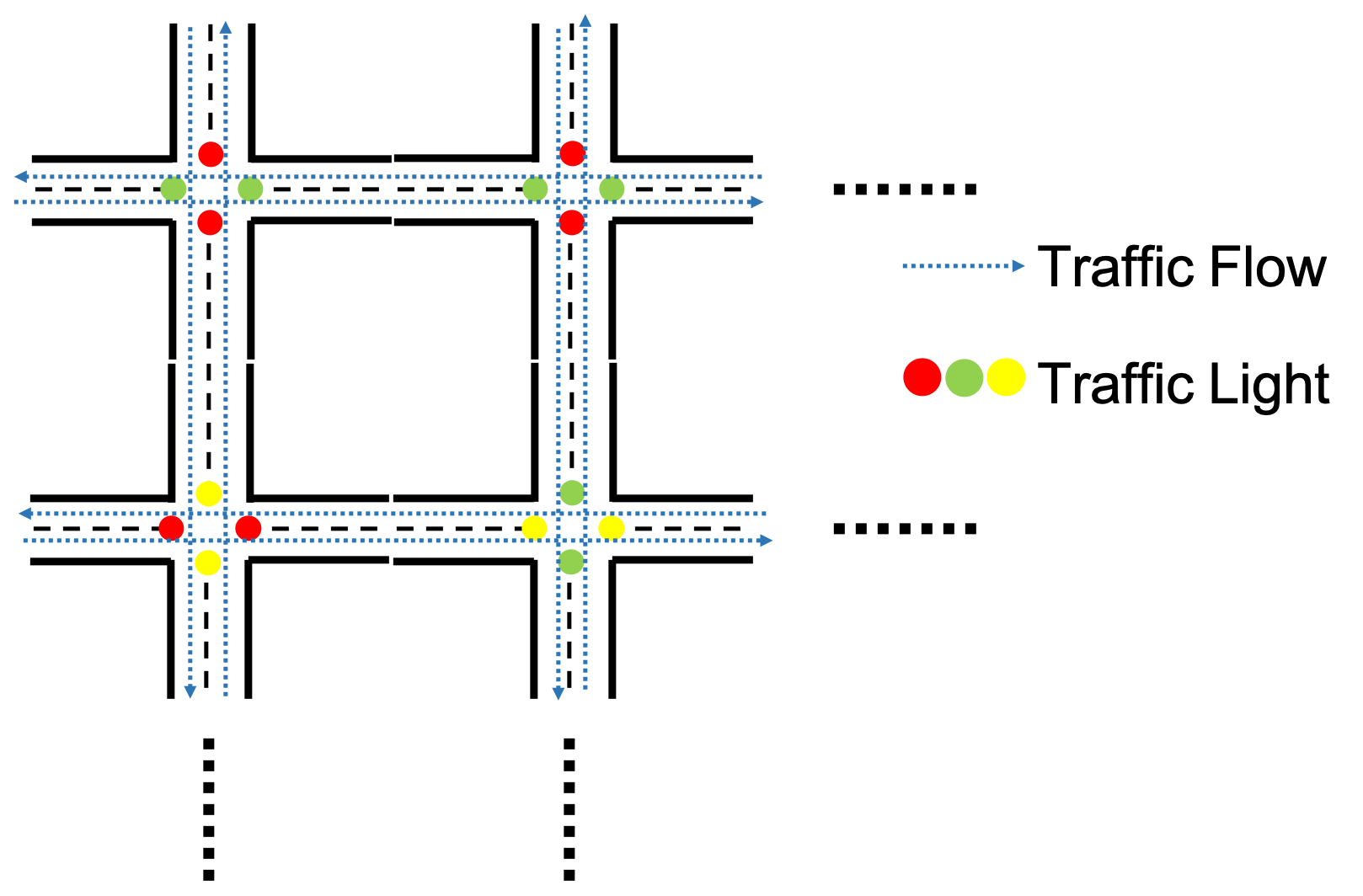} \\
(a) Traffic Light Configuration & (b) Traffic Flow \\[6pt]
\end{tabular}
\caption{The state transition diagram (left) and the traffic flows of a grid topology road network (right). }
\label{fig:DQN}
\end{figure}

Each configuration $Y_n(t)$ can either simply be continued in the next time slot or must otherwise be switched to the natural subsequent configuration $(Y_n(t)+1) \:\: \text{mod}\:\: 4$. This is determined by the action $A(t) \in \{0,1\}$ selected at the end of time slot $t$, which is represented by a binary variable as follows:``0'' for continue, and ``1'' for switch; then we have 
$$Y(t+1) = (Y(t)+A(t)) \: \text{mod} \: 4.$$

These rules give rise to a strictly cyclic control sequence. And the evolution of the queue state over time is governed by the recursion 

$$X_{ni}(t+1) = X_{ni}(t) + C_{ni}(t) - D_{ni}(t),$$
    
with $C_i(t)$ denoting the number of vehicles of traffic flow $i$ appearing at the intersection during time slot $t$ and $D_i(t)$ denoting the number of departing vehicles of traffic flow $i$ crossing the intersection during time slot $t$ and $D_1(t)=0$ if $Y(t) \ne 0$ and $D_2(t)=0 \: \text{if}\: Y(t)\ne 2$.

Traffic flows can either come from outside of the network or from neighbouring intersections belonging to the system. Vehicles coming from outside of the system follow a stochastic process and those from neighbouring intersections are controlled by the states and actions of those intersections.

\subsection{MDP Formulation}
The target Intelligent Transportation System is a Markov Decision Process (MDP). Under the MDP frame work, the state space $\mathcal{S}$, action space $\mathcal{A}$, reward $r$, policy $\pi$, and state transition probability $\mathcal{P}$ of our problem are defined as follows:

\begin{itemize}
    \item State $\mathcal{S} = (X_{ni}(t),Y_n(t))$ where $n = 1,2,\ldots,N$ and $N$ is the number of intersections, $i=1,2,3,4$ are queue length at west, north, east and south directions, respectively. The environment return the full state to the network. However, each agent would only receive local observation, which is $s_n(t)=(X_{ni}(t),Y_n(t)$ for agent $n$.
    \item Action $\mathcal{A} \in \{0,1\}$ where $0$ denotes continue and $1$ denotes switch, as described in traffic model.
    \item Reward $r(s_t,a_t)=- F(X_{ni}(t)),i=1,2,3,4$, where $F(X_{ni}(t))=\sum_{i=1}^4 X_{ni}^2(t)$. The negative sign means that the queue should be penalized. In this paper we use quadratic norm because we want to encourage short queue length. Long queue length would be penalized by the quadratic term of the reward function.
    \item Policy $\pi$ is the signal controlling command from the ITS, which maps a state to a probability distribution over the actions $\pi:\mathcal{S} \to \mathcal{P}(\mathcal{A})$, where $\mathcal{P}(\cdot)$ denotes probability distribution. 
    \item State transition probability $\mathcal{P}$ is the probability of the network entering the new state $s_{t+1}$, after taking the action $a_t$ at the current state $s_t$. At the current state, by taking action $a_t$, vehicles are transmitted to neighbouring intersections or to outside of the system, while stochastic traffic flow from outside of the system during time $t$ to $t+1$ would come into the system.
\end{itemize}

%% file: Sections/DRL.tex
\section{Deep Reinforcement Learning Algorithms}
\label{sect:drl}

In this section, we first present an overview of the DDPG algorithm, then analyze the transportation system. We propose the way to apply spatial influence and reward adjustment \cite{chen2016exploiting} in multi-agent DDPG algorithm in solving signal control problem, and finally describe how they are implemented.

\subsection{Deep Deterministic Policy Gradient Algorithm}

Deep deterministic policy gradient (DDPG) \cite{lillicrap2015continuous} uses deep neural networks to approximate both action policies and value functions (state value or action-value). This method has two advantages. 1) This reduces the dimension of the state space and action space, since it only uses a limited number of parameters to approximate them. 2) Gradient descent can be used to get the optimum, which greatly speeds up the convergence and reduces the computational time. 

While traditional DDPG algorithms have a continuous action space, the traffic control signals in our model are discrete, i.e., $A(t) \in \{0,1\}$. We will apply the discretization process to transform the continuous outputs of the actor network in DDPG to discrete ones. The output layer of actor network has a modified sigmoid function as activation:
    $$
    y=\text{sigmoid}(tx),
    $$
where $x,y$ are inputs and outputs of the final layer, and $t$ is the ratio for steepening the sigmoid function, which is $1000$ in our experiment settings. Combined the modified sigmoid activation function with a node-wise binarization process, our discrete DDPG algorithm can reduce the errors caused by the continuous-to-discrete transformation to a great degree.

\subsection{Intelligent Transportation System Analysis}

\begin{theorem}
\label{theorem:social}
    Consider a road network that has $N$ intersections thus controlled by $N$ agents, the total group network utility can not be improved by single agent changing its action unilaterally. i.e.,
    $$ a_n^* = {\operatorname*{argmax}_{a_n \in \mathcal{A}_n}} \: S_n(a_n,a_{-n}),\forall n\in N, $$
    where $a_{-n}=(x_1,\ldots,x_{n-1},x_{n+1},\ldots,x_N)$ is the set of actions chosen by other agents except agent $n$, $S_n$ is the network group utility, or total traffic congestion cost in our case, and $a_n$ is the action taken by agent $n$.
\end{theorem}

It is worth noting that a traffic system always follows a cooperative game formulation. Also, it is always a total network maximization problem. Optimizing one intersection is of course of great value to the local community. However, to resolve the urban traffic issue in general, that is far from enough. The ITSs need to optimize the road network as a whole to reduce traffic congestion and improve commuting efficiency. 

If we assume that one agent can improve total group utility by changing its action unilaterally in a dynamically stationary road network, that means the vehicles' coming rate and leaving rate are equal for each intersection in a full state transition cycle. If one agent simply changes its strategy, the balance would be broken, queue length at certain directions would increase, thus increasing the total road network utility. 

\begin{theorem}
\label{theorem:nash}
    There is no stationary Nash Equilibrium for Transportation System. 
\end{theorem}

\begin{proof}
Here we prove by contradicting the following two steps.
\begin{itemize}
    \item If there is Nash Equilibrium, it must be $\mathbf{a}=(1,1,\ldots,1)$, if traffic flow from outside of the system is not $0$;
    \item $\mathbf{a}=(1,1,\ldots,1)$ is not a Nash Equilibrium. 
\end{itemize}
To prove Step 1, we only need to assume that there is a $0$ in $\mathbf{a}$. Similar to the proof of Theorem \ref{theorem:social}, if there is a $0$ in $\mathbf{a}$ as the solution of Nash Equilibrium, that means for such intersections, the traffic signal never changes. However, if the traffic signal never changes, the length of $2$ queues at that intersection would be strictly increasing, thus the expected reward would increase for that intersection as well as the traffic network. Therefore, if there is Nash Equilibrium, it must be $\mathbf{a}=(1,1,\ldots,1)$.

To prove Step 2, we first introduce the conclusion of Step 1. Consider a linear road topology with $2$ intersections. So $(X_{1i}(t),Y_1(t))$ and $(X_{2i}(t),Y_2(t)),i = 1,2,3,4$ would be the queue length and light states for agent $1$ and $2$, respectively. And the evolution of the queue state over time is governed by the recursion 
$$
X_{ni}(t+1) = X_{ni}(t) + C_{ni}(t)-D_{ni}(t),
$$
where $n=1,2$, with $C_i(t)$ denoting the number of vehicles of traffic flow $i$ appearing at the intersection during time slot $t$ and $D_i(t)$ denoting the number of departing vehicles of traffic flow $i$ crossing the intersection during time slot $t$, as described in traffic model. Here we make a simple assumption that the number of vehicles appearing at and departing from one intersection at a constant rate $c$ and $d$.
Note that as $C_i(t)$ from outside of the system is a constant, so we have a stationary environment. Also note that $D_i(t)=d$ only when there is green light. Also, as $a_1=a_2=1$ for all $t$, we conclude that $4$ actions complete one cycle. 
If there is Nash Equilibrium, following equation must hold:
$$X_{ni}(t+4) \le X_{ni}(t).$$
Therefore,
\begin{align*}
    X_{ni}(t+4) &\le X_{ni}(t+3)+ C_{ni}(t+3)-D_{ni}(t+3)\\
    &\ldots \\
    &=X_{ni}(t)+4c-d \Rightarrow c \le \frac{1}{4} d.
\end{align*}
If $c=\frac{1}{2}d > \frac{1}{4}d$, the above inequality doesn't hold. So $\mathbf{a}=(1,1,\ldots,1)$ is not Nash Equilibrium for transportation network. 
\end{proof}

However, we can easily infer that if $c=\frac{1}{2}d$, we only need to have $5$ actions in a cycle, to make sure that there are $2$ green light slot in a cycle, thus $X_{ni}(t+5) \le X_{ni}(t)$. 

Note that we used two assumptions in the proof of no Nash Equilibrium: 1. Constant appearing rate of vehicles from outside of the transportation system; 2. Constant departing rate of vehicles at the intersection. These two assumptions are too naive and too strong to be satisfied in real traffic flow. In real life, it is more reasonable to assume that the appearing rate of vehicles is stochastic and the departing rate of vehicles obeys certain probability distribution.

It is counter-intuitive but reasonable to find out that there is no stationary Nash Equilibrium solution for the transportation system. However, it supports the fact that reinforcement Learning algorithm is necessary to learn optimal signal control policy. Otherwise, the hard-coded policy would be enough.

\subsection{Spatial Influence Based Solution}

We propose a unified method for achieving both coordination and communication in Multi-agent reinforcement learning (MARL). For road networks with multiple intersections, we apply multi-agent DDPG in our system. Opposed to the single-agent DDPG method, multi-agent methods do not have computational issues when there are huge observation space and action space. As one agent is only responsible for optimizing traffic signal control for one intersection, it significantly improves the scalability of the system. Also, MADDPG methods would make no assumption about the shape of the network, make modeling mixed direction and shape road networks possible. 

The MADDPG method needs a way to communicate with other agents to compliment the lack of global view for more advanced control like ``greenwave''. Prior work often resorts to centralized training to ensure that agents learn to coordinate. However, centralized training is not feasible in this case, as intersections are not identical. Some intersections are at the border of the network and others are in the middle of the network. Border intersections observe stochastic traffic flows, which means that they need to learn to adapt. Central intersections observe traffic flows coming from other intersections within the network, so they need to cooperate. Inspired by \cite{jaques2019social}, we propose to use social influence as the way to accomplish it. In our work, we make no assumption that agents could view one anothers rewards, as that is often not practical in real life. We relying only on agents viewing each other's actions. We model neighboring agents' actions into the target agent's observation space. The observation space for agent $n$ now changes to $(X_{ni},Y_n,\mathbf{a})$, where $\mathbf{a}$ is the actions taken by neighbouring agents at the last time step. The intuition behind this is that, by observing neighboring agents' actions and the observation space of itself, the agent would be able to learn the policy of traffic signal control and balance between optimizing the reward of itself, its neighbourhood and indirectly the whole system.

\subsection{Rewards Adjustment}
As we have shown in the previous option, total group network utility can not be improved by a single agent changing its action unilaterally. Therefore, we propose to adjust the reward function of the agent to make sure that each agent would not only consider its own reward, but also the rewards of its neighbouring agents. We change the reward function from:
$$ r_n(s_t,a_t)=-F(X_n(t))$$
to
\begin{equation}
\label{eqt:reward}
    r_n(s_t,a_t)=-F(X_n(t))-\sum_{m \in \mathcal{M}_n}w_{nm}F(X_m(t)) ,
\end{equation} 
where $\mathcal{M}_n$ is the neighbouring agents of agent $n$, $w_{nm}$ is the social tie weight. We define $\sum_{m \in \mathcal{M}}w_{nm}$ as ``Selfish Index'' as it determines the importance of neighbouring agents to the target agent. Based on this definition, the agent would not only take its own reward into account, but its neighbourhood as well. By directly linking agent with its neighbourhood and thus indirectly with agents far away, the agent is expected to learn to cooperate and communicate with each other to optimize the social group total utility. 

\subsection{DDPG for Grid Road Networks}
For road networks with multiple intersections, we apply the multi-agent DDPG algorithm in the network. We have one agent for each intersection. The number of cars and lights states at all intersections are inputs to the MADDPG system. To ensure that the MADDPG can possess a non-local view regarding the total network utility, spatial influence and social group utility methods are considered, as we have specified in the previous sections. The algorithm is defined in Algorithm \ref{alg:DDPG}: 

\begin{algorithm}
    \caption{Multi-agent DDPG for ITSs}
    \KwInput{number of episodes $M$, time frame $T$, minibatch size $N$, learning rate $\lambda$, and number of agents or intersections $J$}
    \begin{algorithmic}[1]
    \label{alg:DDPG}
        \FOR{$j = 1, J$}
            \STATE Randomly initialize critic network $Q_j(O_j,a|\theta_j^Q)$ and actor network $\mu_j(O_j|\theta_j^\mu)$ with random weight $\theta_j^Q$ and $\theta_j^\mu$ for agent $j$;
            \STATE Initialize target network $Q'_j$ and $\mu'_j$ with weights $\theta_j^{Q'} \leftarrow \theta_j^{Q}$, $\theta_j^{\mu'} \leftarrow \theta_j^{\mu}$ for each agent $j$;
            \STATE Initialize replay buffer $B_j$ for each agent $j$;
        \ENDFOR
        \FOR {episode $= 1, M$}
            \STATE Initialize a random process $\mathcal{N}$ for action exploration;
            \STATE Receive initial observation state $s_0$;
            \FOR{$t = 1, T$}
                \FOR{$j = 1, J$}
                    \STATE {Select action $a_{j,t} = \mu_j(O_{j,t}|\theta_j^\mu) + \mathcal{N}_t$ according to the current policy and exploration noise;}
                \ENDFOR
                \STATE {Each agent executes action $a_{j,t}$, market state changes to $s_{t+1}$;}
                \STATE {Each agent observes reward $r_{j,t}$ and observation $O_{j,t+1}$, where observation $O_{j,t+1}$ is adjusted according to spatial influence};
                \STATE {Each agent adjust reward according to Equation \ref{eqt:reward}}
                \FOR{$j = 1, J$}
                    \STATE {Store transition ($O_{j,t}$, $a_{j,t}$, $r_{j,t}$, $O_{j,t+1})$ in $B_j$;}
                    \STATE {Sample a random minibatch of $N$ transitions ($O_{j,i}$ , $a_{j,i}$ , $r_{j,i}$ , $O_{j,i+1}$) from $B_j$;}
                    \STATE {Set $y_{j,i} = r_{j,i}+\gamma Q'_j (s_{t+1}, \mu'_j (O_{j,i+1}|\theta_j^{\mu'}|\theta_j^{Q'}))$
                    for $i = 1, \ldots, N$;}
                    \STATE {Update the critic by minimizing the loss: $L = \frac{1}{N}\sum_i(y_{j,i} -Q_j(O_{j,i},a_{j,i}|\theta_j^Q))^2$;}
                    \STATE {Update the actor policy by using the sampled policy gradient:
                        \begin{multline*}
                            \nabla_{\theta^\mu} \pi \approx \frac{1}{N}\sum_i \nabla_a Q_j(O,a|\theta_j^Q)|_{O = O_{j,i},a = \mu_j(O_{j,i})}
                            \times \nabla_{\theta^\mu} \mu_j(O_j|\theta^\mu)|_{s_i};
                        \end{multline*}
                    }
                    \STATE {Update the target networks:\:\:$\theta_j^{Q'}\leftarrow \tau \theta_j^Q + (1-\tau)\theta_j^{Q'},\:\:\:\:\theta_j^{\mu'} \leftarrow \tau \theta_j^\mu + (1-\tau)\theta_j^{\mu'}.$}
                \ENDFOR
            \ENDFOR
        \ENDFOR
    \end{algorithmic}
\end{algorithm}

%% file: Sections/Performance.tex
\section{Performance Evaluation}
\label{sect:performance}

In this section, we discuss the numerical experiments implemented to evaluate the performance of the social influence based MADDPG algorithms as detailed above. In the experiments, we apply a 4-layer fully-connected neural network for both actor and critic in both the linear topology and grid scenario, with 400 neurons for the first 2 layers, following with two layers with 600 and 200 neurons, respectively. All agents share the same architecture, while they have difference observation space, as spatial influence is applied to them. Noting that the actor needs to output near-binary values as action values as mentioned above, the actor neural network has one modified sigmoid activation function for the output layer. As usual, a copy of the actor and critic neural network are taken as the target networks with ``soft'' updates. The outputs of the actor network are clipped to binary values 0 and 1 indicating light state changing and remaining the same. We take an episode length of 150 steps of simulation for collecting the learning samples, and train both actor and critic networks with a batch size of 64 and discount factor $\gamma =0.99$. The OrnsteinUhlenbeck noise \cite{horsthemke1980perturbation} is applied for the explorations, with a variance of 0.3.

As for the traffic environment settings in the experiments, the vehicle coming and passing rates are set differently for main roads and branch roads. the vehicle coming rate indicates the number of vehicles coming from outside of the road networks into the networks on each road per time step, which is set to be a random value with upper bound $C_m$ for main roads or $C_b$ for branch roads; and the vehicle passing rates indicate the number of vehicles passing one intersection within one time step, set to be the fixed numbers of 16 and 4 for main road and branch road, respectively.

\subsection{Grid Road Topology}
We assume the numbers of arriving vehicles in eastern and western direction to the arterial road in each time step to be independent and Bernoulli distributed with parameter $p_1$. The numbers of arriving vehicles in southern and northern directions in each time step are also independent and Bernoulli distributed with parameter $p_2$.
We extend the road network to be a grid topology. The intersection setting is described in Section \ref{sect:problem} and the road network has more than one intersections in south-north and east-west directions. The quadratic congestion cost function is of the form $F(X)=\sum_{n=1}^N \sum_{i=1}^4 X_{ni}^2$. In our experiments, we design a road network with three columns and three rows, therefore $9$ intersections in total. Fig. \ref{fig:directions} shows that the algorithm with spatial influence and reward adjustment integrated performs the best, no matter the directions of the influencing flow. 
\begin{figure}
\begin{tabular}{ccc}
\centering
  \includegraphics[width=45mm]{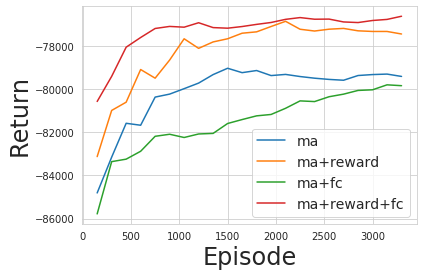} &   \includegraphics[width=45mm]{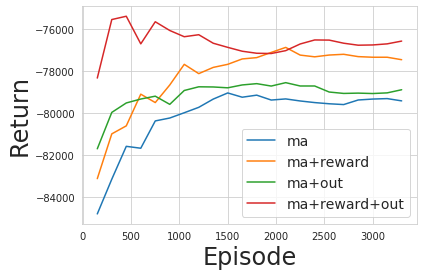}&
  \includegraphics[width=45mm]{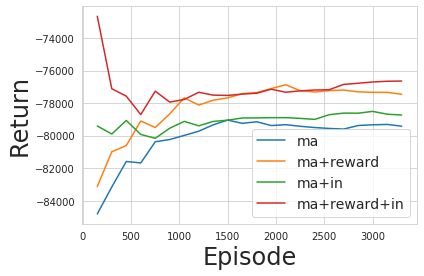}\\
(a) Fully Connected & (b) Outward & (c) Inward \\[6pt]
\end{tabular}
\caption{Different influencing flow. We can see for all of the influencing direction flow, by adding spatial influence and reward adjustment functions, the traffic system would have higher return.}
\label{fig:directions}
\end{figure}

\subsubsection{Spatial Influence}
As we have described in Section \ref{sect:drl}, we design spatial influence flow to affect the behavior of neighbouring agents. In inward and outward scenarios, for any two agents that are physically connected, there will be only one influencer and one influencee. For a fully connected influence scenario, all agents can observe the last actions from all its neighbours. 
\begin{itemize}
    \item Inward Influence: agents at the four corners are independent.  
    \item Outward Influence: agent at the center is independent.
    \item Fully Connected Influence: all agents are not independent.
\end{itemize}

As we can see from Fig. \ref{fig:directions_selfish}, the total utility converges to the same level. So all the methods can optimize total network utility. However, we can see from the detailed analysis of individual utility that in inward influence flow, agents at the corners have less congestion compared to other scenarios, while in outward flow scenario, the agent at the center has relatively less congestion. Therefore, in a city where the importance of interactions varies, the direction of influencing flow should be considered. Thus the city traffic flow can be optimized both globally and locally. 

\begin{figure}
\begin{tabular}{cc}
\centering
  \includegraphics[width=50mm]{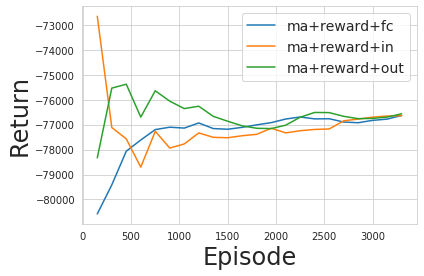} &   \includegraphics[width=50mm]{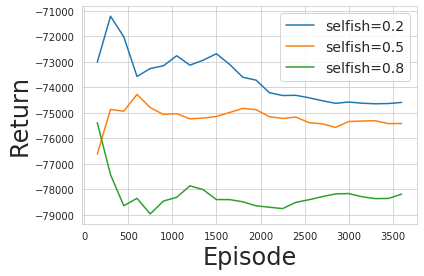} \\
(a) Directions & (b) Selfish index \\[6pt]
\end{tabular}
\caption{The influencing directions do not affect the performance level at the convergence (left). The ``Selfish Index'' affects the system performance (right). If agents are cooperative, the total group utility would be better. }
\label{fig:directions_selfish}
\end{figure}

\subsubsection{Social Group Utility}
We adjust the reward function of each agent according to Equation (\ref{eqt:reward}). As we can see from Fig. \ref{fig:directions_selfish}, the total network utility increases as the social tie weights of neighbouring agents increases. This confirms our anticipation that as neighbouring agents' rewards become more important to the target agent, it will learn to cooperate with its neighbours to increase the road network utility.

%% file: Sections/Conclusion.tex
\section{Conclusion}
\label{sect:conclusion}

We have explored the scope of using social influence based Multi-agent Deep Deterministic Policy Gradient (MADDPG) to optimize real-time traffic signal control policies in emerging large-scale Intelligent Transportation Systems. We compared the performance between social influence based MADDPG and naive MADDPG, and demonstrated that social influence based algorithm has the maximum potential in scaling without the issue of large observation space and action space and would help communication between agents. We verified it the scalability properties of DDPG algorithms in both a linear topology and a grid topology, and demonstrated the emergence of intelligent behavior such as ``green wave'' patterns, confirming that the agents learned to cooperate through social influence communication.